\documentclass[12pt]{article}
\usepackage{tikz}

\usepackage{amsfonts}
\usepackage{amsmath}
\usepackage{amsthm}
\usepackage{uhrzeit}
\usepackage{graphicx}
\usepackage{epstopdf}
\usepackage{hyperref}
\usepackage{mathtools}

\title{Slightly Improved Upper Bound on the Integrality Ratio for the $s-t$ Path TSP}
\author{Xianghui Zhong\\
\small University of Bonn, Germany\\[5mm]
            }
\date{\today} 

\newtheorem{theorem}{Theorem}[section]
\newtheorem{lemma}[theorem]{Lemma}

\newtheorem{remark}[theorem]{Remark}
\theoremstyle{definition}
\newtheorem{definition}[theorem]{Definition}

\begin{document}
\maketitle
\begin{abstract}
In this paper we investigate the integrality ratio of the standard LP relaxation for the \textsc{Metric $s-t$ Path TSP}. We make a near-optimal choice for an auxiliary function used in the analysis of Traub and Vygen which leads to an improved upper bound on the integrality ratio of 1.5273.
\end{abstract}

{\small\textbf{keywords:} traveling salesman problem; metric $s-t$ path TSP; integrality ratio}

\section{Introduction}
The traveling salesman problem (\textsc{TSP}) is probably the best-known problem in discrete optimization. An instance consists of a complete graph $K_n$ with a distance function $c$ on the edges of the graph and the task is to find a shortest Hamilton cycle, i.e.\ a tour visiting every vertex exactly once. 

A variant of the \textsc{TSP} is the \textsc{$s-t$ Path TSP} where two vertices $s$ and $t$ are specified and the task is to find a shortest path starting in $s$ and ending in $t$ visiting all other vertices exactly once. In this paper we consider the \textsc{Metric $s-t$ Path TSP} as a special case of it. Here the distances satisfy the triangle inequality. By recent research it is known that this problem can be approximated within a factor of $\frac{3}{2}+\epsilon$ and $\frac{3}{2}$ \cite{traub2018approaching,zenklusen20191}. Moreover, it was shown that any $\alpha$ approximation algorithm for the standard \textsc{TSP} problem implies an $\alpha+\epsilon$ approximation algorithm in the \textsc{$s-t$ Path TSP} version \cite{TraubVZ20}.

The \emph{integrality ratio} of a linear program (LP) is the supremum of the ratio between the value of the optimal integral solution and that of the optimal fractional solution. In other words if $OPT(I)$ and $OPT_{LP}(I)$ are the values of the optimal integral solution and optimal fractional solution of an instance $I$, then the integrality ratio is defined as $\sup_I \frac{OPT(I)}{OPT_{LP}(I)}$.

An interesting open question asks for the integrality ratio of the following standard LP relaxation of the problem:
\begin{align*}
\min \sum_{e\in E(K_n)} &c_ex_e\\
\sum_{e \in \delta(v)} x_e &=2 && \text{for all~} v \in V(K_n)\backslash \{s,t\} \\
\sum_{e \in \delta(v)} x_e &=1 && \text{for all~} v \in \{s,t\}\\
\sum_{e \in \delta(X)}x_e &\geq 2 && \text{for all~} X \subseteq V(K_n)\backslash \{s,t\}\\
\sum_{e \in \delta(X)}x_e &\geq 1 && \text{for all~} \{s\}\subseteq X \subseteq V(K_n)\backslash \{t\}\\
x_e&\geq 0 && \text{for all~} e\in E(K_n).
\end{align*}
The best currently known lower bound is $\frac{3}{2}$. This value is achieved by a simple standard example. A recent series of work improves the upper bound towards the conjectured optimal value of $\frac{3}{2}$.

Hoogeveen adapted Christofides algorithm for the standard TSP \cite{christofides} (which was independently developed by Serdjukov \cite{serdjukov}) to the $s-t$ \textsc{Path TSP} \cite{hoogeveen1991analysis}. A parity correction vector is added to a minimum spanning tree to obtain a tour. This leads to an integrality ratio of $\frac{5}{3}$ for the path version. An, Kleinberg and Shmoys suggested the best of many Christofides algorithm for $s-t$ \textsc{Path TSP} \cite{an2015improving}. Instead of using the minimum spanning tree they decompose the optimal LP solution into a convex combination of spanning trees. Then, they sample the spanning tree according to the convex combination, add a parity correction vector and output the best result. With this approach the upper bound on the integrality ratio was improved to $\frac{1+\sqrt{5}}{2}$. Seb\H{o} improved and simplified this approach to obtain a ratio of $\frac{8}{5}$ \cite{sebHo2013eight}. In \cite{vygen2016reassembling} Vygen choose the convex combination in a particular way. This idea was further improved by Gottschalk and Vygen by a generalization of the Gao trees \cite{gottschalk2018better}.
For an upper bound of $\frac{3}{2}+\frac{1}{34}$, Seb\H{o} and Van Zuylen delete the so-called lonely edges of the spanning trees before adding the parity correction vector based on the underlying idea that it is likely that the parity correction vector will reconnect the tour. If this is not the case they add two copies of lonely edges to reconnect the tour afterwards \cite{sebHo2019salesman}. The analysis was improved by Traub and Vygen by choosing the weights of the spanning trees in a non-standard way. This improves the ratio to $1+\frac{1}{1+4\ln(\frac{5}{4})}$ \cite{traub2019improved}.

In the special case of the graph version of the $s-t$ \textsc{Path TSP} where the cost $c$ arises from shortest paths of a unweighted graph $G$ the integrality ratio is known to be equal to $\frac{3}{2}$ \cite{sebHo2014shorter}. 

In this paper we improve the previous best upper bound on the integrality ratio for the \textsc{Metric $s-t$ Path TSP} of $1+\frac{1}{1+4\ln(\frac{5}{4})}>1.5283$ by Traub and Vygen in \cite{traub2019improved} to 1.5273 by improving the choice of an auxiliary function $h$ in their analysis. Numerical computations indicate that this bound is close to the best possible upper bound achievable by their approach.

\section{Improved Upper Bound on the Integrality Ratio}
In this section we show the improved upper bound on the integrality ratio. For this we use a theorem from \cite{traub2019improved}:

\begin{theorem}[Theorem 5 in \cite{traub2019improved}] \label{integrable fkt}
Let $h:[0,1]\to [0,1]$ be an integrable function with
\begin{align}
\int_z^1\max\{0, h(\sigma)-1 +zh(\sigma)\}\mathrm{d}\sigma+\int^z_0(h(\sigma)-1-zh(\sigma))\mathrm{d}\sigma\leq 0 \label{cond}
\end{align}
for all $z\in [0,1]$. Then, the best-of-many Christofides algorithm with lonely edge deletion \cite{sebHo2019salesman} computes a solution of cost at most $\rho^*c(x^*)$, where
\begin{align*}
\rho^*=1+\frac{1}{1+\int_0^1h(\sigma)\mathrm{d}\sigma}.
\end{align*}
\end{theorem}

Traub and Vygen applied Theorem \ref{integrable fkt} for $h(\sigma):=\frac{4}{4+\sigma}$. We define our choice of $h$ as follows:

\begin{definition}
$h$ is a step function taking the value $\alpha:=0.971239$ in $[0,x)$ and the value $\beta:=0.873362$ in $[x,1]$ where $x:=0.236901$ (Figure \ref{choice of h}), i.e.\
\begin{align*}
h(\sigma)=  \begin{cases*}
      0.971239 & if $\sigma\in[0,0.236901)$ \\
      0.873362        & otherwise.
    \end{cases*}
\end{align*}
\end{definition}

\begin{figure}[!htb]
\centering
 \includegraphics[scale=0.5]{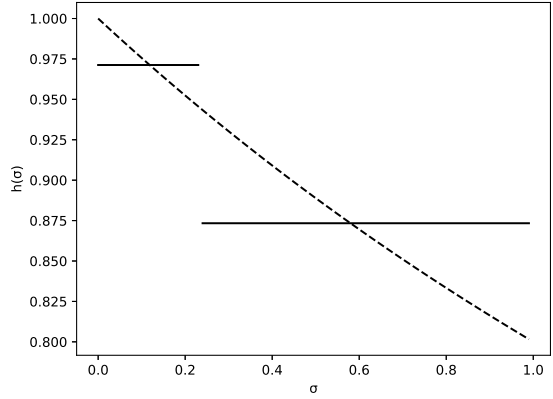}
  \caption{The straight function shows our choice of $h$, the dashed function that chosen by Traub and Vygen in \cite{traub2019improved}.}
  \label{choice of h}
\end{figure}

In order to apply Theorem \ref{integrable fkt} we need to show that the condition (\ref{cond}) is satisfied.

\begin{lemma} \label{condition}
Our choice of $h$ satisfies the condition of Theorem \ref{integrable fkt}.
\end{lemma}

\begin{proof}
Since by definition $h(\sigma)>0$ for all $\sigma\in[0,1]$ we have $h(\sigma)-1 +zh(\sigma)<0$ if and only if $z< \frac{1}{h(\sigma)}-1$. For our choice of $h$ we have that $h(\sigma)$ can only take two values: $\alpha$ and $\beta$. Note that $0<\frac{1}{\alpha}-1 <0.03$ and $0.145<\frac{1}{\beta}-1 <0.146$. Thus, we can distinguish four cases: $z\in [0,\frac{1}{\alpha}-1), z\in [\frac{1}{\alpha}-1,\frac{1}{\beta}-1), z\in [\frac{1}{\beta}-1,x)$ and $z\in [x,1]$. 

The first case is trivial, since for $z\in[0,\frac{1}{\alpha}-1)$ we have:
\begin{align*}
&\int_z^1\max\{0, h(\sigma)-1 +zh(\sigma)\}\mathrm{d}\sigma+\int^z_0(h(\sigma)-1-zh(\sigma))\mathrm{d}\sigma\\
=&\int^z_0((1-z)h(\sigma)-1)\mathrm{d}\sigma \leq \int^z_0 0 \mathrm{d}\sigma=0.
\end{align*}

For $z\in [\frac{1}{\alpha}-1,\frac{1}{\beta}-1)$ we have:
\begin{align*}
&\int_z^1\max\{0, h(\sigma)-1 +zh(\sigma)\}\mathrm{d}\sigma+\int^z_0(h(\sigma)-1-zh(\sigma))\mathrm{d}\sigma\\
=&\int_z^{x} (\alpha-1 +z\alpha) \mathrm{d}\sigma+z(\alpha-1-z\alpha)=(x-z) (\alpha-1+z\alpha)+z(\alpha-1-z\alpha)\\
=&x\alpha-x+x\alpha z-2\alpha z^2.
\end{align*}
Similarly, for $z\in [\frac{1}{\beta}-1,x)$ we have:
\begin{align*}
&\int_z^1\max\{0, h(\sigma)-1 +zh(\sigma)\}\mathrm{d}\sigma+\int^z_0(h(\sigma)-1-zh(\sigma))\mathrm{d}\sigma\\
=&\int_z^{x} (\alpha-1 +z\alpha) \mathrm{d}\sigma+\int_x^{1} (\beta-1 +z\beta) \mathrm{d}\sigma+z(\alpha-1-z\alpha)\\
=&(x-z) (\alpha-1+z\alpha)+(1-x)(\beta-1+z\beta)+z(\alpha-1-z\alpha)\\
=&x\alpha-1+(1-x)\beta+(x\alpha+(1-x)\beta)z-2\alpha z^2.
\end{align*}
Finally, for $z\in [x,1]$ we have:
\begin{align*}
&\int_z^1\max\{0, h(\sigma)-1 +zh(\sigma)\}\mathrm{d}\sigma+\int^z_0(h(\sigma)-1-zh(\sigma))\mathrm{d}\sigma\\
=&\int_z^{1} (\beta-1 +z\beta) \mathrm{d}\sigma+\int^x_0(\alpha-1-z\alpha)\mathrm{d}\sigma+\int^z_x(\beta-1-z\beta)\mathrm{d}\sigma\\
=&(1-z)  (\beta-1+z\beta)+x(\alpha-1-z\alpha)+(z-x)(\beta-1-z\beta)\\
=&x\alpha-1+(1-x)\beta+(-x\alpha+(x+1)\beta)z-2\beta z^2.
\end{align*}
Hence, it is enough to show that for all $z\in \mathbb{R}$ we have:
\begin{align*}
x\alpha-x+x\alpha z-2\alpha z^2&\leq 0\\
x\alpha-1+(1-x)\beta+(x\alpha+(1-x)\beta)z-2\alpha z^2&\leq 0\\
x\alpha-1+(1-x)\beta+(-x\alpha+(x+1)\beta)z-2\beta z^2&\leq 0.
\end{align*}
The left hand sides are quadratic functions in $z$. Note that the leading coefficient is negative in all three cases. Hence, the inequalities hold if and only if the discriminants of the three quadratic functions are non-positive, that is:
\begin{align}
(x\alpha)^2+8\alpha(x\alpha-x)&\leq 0 \label{eq1}\\
(x\alpha+(1-x)\beta)^2+8\alpha(x\alpha-1+(1-x)\beta) &\leq 0 \label{eq2} \\
(-x\alpha+(x+1)\beta)^2+8\beta(x\alpha-1+(1-x)\beta)&\leq 0 \label{eq3}.
\end{align}
We can check these inequalities for our choice of $x, \alpha, \beta$:
\begin{align*}
(x\alpha)^2+8\alpha(x\alpha-x)<-1.17266 \cdot 10^{-7} &< 0\\
(x\alpha+(1-x)\beta)^2+8\alpha(x\alpha-1+(1-x)\beta)<-3.5346\cdot 10^{-6} &< 0\\
(-x\alpha+(x+1)\beta)^2+8\beta(x\alpha-1+(1-x)\beta)<-3.00596\cdot 10^{-6} &< 0.
\end{align*}
\end{proof}

\begin{theorem}
The integrality ratio of the standard LP relaxation for the \textsc{Metric $s-t$ Path TSP} is at most 1.5273.
\end{theorem}

\begin{proof}
By Lemma \ref{condition}, our choice of $h$ satisfies the condition of Theorem \ref{integrable fkt}. Hence, we can apply Theorem \ref{integrable fkt} to get an upper bound on the integrality ratio of
\begin{align*}
\rho^*:=1+\frac{1}{1+\int_0^1h(\sigma)\mathrm{d}\sigma}=1+\frac{1}{1+x\alpha+(1-x)\beta}<1.5273.
\end{align*}
\end{proof}

\begin{remark}
The values for $\alpha, \beta$ and $x$ we chose are approximate values of a solution for the system of equations we get by replacing the less-than-or-equal sign in (\ref{eq1}), (\ref{eq2}) and (\ref{eq3}) by an equal sign. More precise values would probably lead to a better upper bound. Using a computer algebra system, we can solve that system of equations to get the exact values:

\begin{align*}
\alpha&:=\frac{1}{48}\left(34+\frac{73}{\sqrt[3]{-377+18i\sqrt{762}}}+\sqrt[3]{-377+18i\sqrt{762}} \right)\\
\beta&:=\frac{2}{3}\left(-45+172\alpha-128\alpha^2\right)\\
x&:=8(\frac{1}{\alpha}-1),
\end{align*}
where the roots are principal roots and $i$ is the imaginary unit. By definition, it is clear that the inequalities (\ref{eq1}), (\ref{eq2}) and (\ref{eq3}) are fulfilled by this choice of values. This would lead to an upper bound on the integrality ratio of
\scriptsize
\begin{align*}
\frac{-30 (377 i + 18 \sqrt{762}) + (-377 + 18 i \sqrt{762})^{\frac{2}{3}} (-249 i + 28 \sqrt{762}) + (-377 + 18 i \sqrt{762})^{\frac{1}{3}} (-3975 i + 206 \sqrt{762})}{4((-377+18i\sqrt{762})^{\frac{2}{3}} (-44 i + 7 \sqrt{762}) - 16 (377 i + 18 \sqrt{762}) + (-377 + 18 i \sqrt{762})^{\frac{1}{3}} (-1088 i + 47 \sqrt{762}))}<1.5273.
\end{align*}
\end{remark}

\begin{remark}
As already pointed out in \cite{traub2019improved} numerical computations indicate that the best choice of $h$ gives an upper bound on the integrality ratio of approximately 1.5273. Hence, this suggests that our choice of $h$ is near-optimal.
\end{remark}

\end{document}